\documentclass[preprint,10pt]{elsarticle}
\usepackage{color}
\makeatletter
\def\ps@pprintTitle{%
 \let\@oddhead\@empty
 \let\@evenhead\@empty
 \let\@oddfoot\@empty
 \let\@evenfoot\@empty
}
\makeatother

\usepackage{amssymb}
\usepackage{amsmath}
\usepackage{array}
\usepackage{cases}
\usepackage{helvet}
\usepackage{multirow}
\usepackage{amsmath, amsfonts, amssymb}
\usepackage{amsthm}
\usepackage[english]{babel}
\usepackage{graphicx}
\usepackage{url}
\usepackage{bm}
\usepackage{enumitem} 
\usepackage{multirow}
\usepackage{booktabs}
\usepackage{algorithm}
\usepackage{algorithmic}
\usepackage{booktabs}
\usepackage{threeparttable}
\usepackage{graphicx}%导入图片
\usepackage{subcaption} %让跨双栏图片出现在正确的位置

\usepackage{natbib}
\usepackage{epstopdf}
\usepackage{setspace}

\usepackage[figuresright]{rotating}
\usepackage{geometry}
\usepackage{color}%
\usepackage{float}
\numberwithin{equation}{section}

\begin{document}
	
\newtheorem{definition}{$\mathbf{Definition}$}[section]
\newtheorem{theorem}{$\mathbf{Theorem}$}[section]
\newtheorem{corollary}{$\mathbf{Corollary}$}[section]
\newtheorem{lemma}{$\mathbf{Lemma}$}[section]
\newtheorem{proposition}{$\mathbf{Proposition}$}[section]
\newtheorem{property}{$\mathbf{Property}$}[section]
\newtheorem{remark}{$\mathbf{Remark}$}[section]
\newtheorem{example}{$\mathbf{Example}$}[section]

\begin{frontmatter}
\title{Energy Detection for Cognitive Radio with Distributional Uncertainty and Signal Variety under Nonlinear Expectation Theory}
  
\author[add1]{Jialiang Fu} \ead{fujialiang@amss.ac.cn}
\author[add1,add2]{Wen-Xuan Lang\corref{cor1}} \ead{langwx@amss.ac.cn}
		
\address[add1]{Academy of Mathematics and Systems Science, Chinese Academy of Sciences, Beijing 100190, China }
\address[add2]{National Center for Mathematics and Interdisciplinary Sciences, Chinese Academy of Sciences, Beijing 100190, China}

\cortext[cor1]{Corresponding author}

\nonumnote{This work is supported by National Key R\&D Program of China under Grant 2023YFA1009604.}
				
\begin{abstract}
\noindent			
Classical energy detection (ED) methods for cognitive radio (CR) have addressed noise uncertainty as deviations in noise power and signal uncertainty as variability in signal characteristics, which use probabilistic methods and assume fixed probability distributions for both. In practical scenarios, due to the uncertainty in probability models and the significant variation of primary signals encountered by receivers across different radio technologies, wireless environments exhibit not only distributional uncertainty but also substantial signal variety. In this paper, we develop a generalized formulation of energy detection based on nonlinear expectation theory, where both the signal and noise distributions are uncertain. We utilize the $G$-normal distribution to characterize channel noise. Moreover, to capture practical signal variety, the absolute values of transmitted signal random variables are assumed to lie within a bounded range $[\underline{\sigma}_X,\overline{\sigma}_X]$. The worst-case detection performance is then characterized by a double supremum, meaning over all admissible distributions and all possible signal realizations. We derive estimations for the minimum and maximum detection error probabilities, and demonstrate the validity of the results through numerical simulations. The proposed model generalizes the classical theoretical analysis of energy detection and offers a potential theoretical foundation for robust detection and information-theoretic analysis under distributional uncertainty.

\end{abstract}
		
		
\end{frontmatter}
\date{}

\section{Introduction}\label{sec1}
Cognitive radio (CR) is a communications paradigm \cite{JM1999cognitive} that can effectively address the problem of spectrum insufficiency. This technology allows systems to dynamically optimize the use of available frequency bands while ensuring less interference with licensed bands \cite{AN2021}. Energy detection (ED) is a fundamental approach \cite{HU1967energy} for spectrum sensing in CR to improve spectrum usage efficiency, where unlicensed (secondary) users must identify the presence or absence of licensed (primary) signals in a shared frequency spectrum. ED has been extensively studied in recent years, with numerous contributions focusing on its performance, optimization, and application in various CR scenarios \cite{FF2007on,SS2013on,YY2019improved,XZ2019signal,NI2020generalized}. This approach has become a widely adopted spectrum sensing technique for CR, owing to its versatility, simplicity, and low computational and implementation costs.

In classical information theory, probability theory provides the foundation, and the classical ED theory assumes that both signal and noise follow deterministic distributions, such as the Gaussian distribution. However, in practical wireless communication systems, imperfect calibration of system noise and the inability to know the signal's characteristics in advance can arise as significant challenges. To address these challenges, \cite{RT2008snr,RT2007snr} investigated the noise uncertainty, which refers to the inability to perfectly know the system noise properties to infinite precision, and modeled this uncertainty by combining the nominal and estimated noise power. \cite{ML2010performance,ML2013signal} introduced the signal uncertainty, which refers to the inability of a user to perfectly know beforehand the primary signals that might be present in the sensed band and their properties, and used a modified Gaussian model to characterize the received Signal-to-Noise Ratio (SNR) distribution. While these studies offer solutions, the common assumption is that there exists a particular distribution characterizing the random variable, despite uncertainty in the specific parameters and/or the form of this distribution. In other words, the existing methods still assume that the underlying probability model (i.e., the probability space) is deterministic.

However, in the complex physical world, the assumption of using precise and well-defined probability distributions to describe uncertainty is somewhat idealized. As mentioned in \cite{spiegelhalter2024does}, probability is not an objective property of the world but rather a construct based on subjective judgments and assumptions. In many cases, such as when there are limited observations or small sample sizes, there may not be a single, fixed probability distribution that accurately captures the underlying randomness. This reality challenges the traditional approach of relying on deterministic probability models, which typically assume the existence of an underlying ``true'' probability distribution. This limitation motivates the need for alternative approaches, such as nonlinear expectation theory, which can model uncertainty in a more flexible and robust manner, better handling incomplete or imprecise information.

In various communication environments, the channel noise may change its mean and variance due to changing transmission circumstances \cite{AJ2022snr}, and the detection and estimation of such changes are crucial for maintaining the stability and performance of communication systems. Moreover, in practical CR scenarios, the primary signals differ significantly in their transmission dynamics, power levels, and statistical properties. Previous works have provided excellent models for various signal types \cite{ML2013signal}, such as digital TV, DAB-T, and E-GSM, accurately capturing their characteristics. However, these models are typically based on deterministic probability distributions and require prior knowledge of the signal type. While these approaches are highly effective for specific signal types, a more general analytical method that does not rely on signal type identification is still needed. These needs, coupled with distributional uncertainty of both the signal and noise, motivate the need for a more general theoretical analysis framework for energy detection, one that can accommodate the inherent distributional uncertainty and diversity of signals in real-world wireless environments.

Nonlinear expectation theory \cite{peng2007G,peng2019nonlinear} is an emerging field of research that extends classical probability theory by using sublinear expectations instead of linear expectations, providing a powerful framework for effectively addressing the uncertainty of probability models. By incorporating the distributional uncertainty, this theory provides a more robust framework for modeling complex systems. One important result of the theory is the central limit theorem under sublinear expectations \cite{peng2008new}. This theorem indicates that the $G$-normal distribution is a highly universal distribution and widely applicable, and can be regarded as a generalization of the Gaussian distribution in probability theory. In \cite{WX2025hypothesis}, the authors formulated a hypothesis testing model for channel noise with distributional uncertainty and modeled the channel noise as the $G$-normal distribution, which is used to describe the uncertainty associated with the variance of the noise. This highlights the feasibility of employing nonlinear expectation theory to model and address such distributional uncertainty.

In this paper, we develop a generalized formulation of energy detection based on nonlinear expectation theory, where both the signal and noise distributions are uncertain. We model the noise as a $G$-normal distribution. The transmitted signal is assumed to be uncertain in both distribution and type, with only the knowledge that its absolute values take values within a bounded interval $[\underline{\sigma}_X, \overline{\sigma}_X]$, reflecting the uncertainty in both the statistical properties of the signal and its actual transmission levels. We still model the problem within the framework of binary hypothesis testing, with the detection performance characterized by taking the supremum over all admissible distributions, and the supremum or infimum over all possible realizations of the signals. Leveraging the properties of the $G$-normal distribution, we estimate the maximum and the minimum of the upper probability of missed detections, and the upper probability of false alarms. Numerical results show that our conclusions can serve as upper and lower bounds for results derived using probabilistic methods, providing a general theoretical reference that does not require knowledge of the signal type or detection noise characteristics.

This paper is organized as follows. In Section \ref{sec2}, we briefly introduce the fundamental concepts of nonlinear expectation theory. In Section \ref{sec3}, we establish the required hypothesis testing model with $G$-normal noise and uncertain input signals. Section \ref{sec4} is devoted to deriving the maximum and the minimum of the upper probability of missed detections and the upper probability of false alarms, and discusses the results under several types of fading. In Section \ref{sec5}, we present numerical experiments to compare the results of this paper with those of traditional results. Finally, we draw conclusions in Section \ref{sec6}.

\section{Preliminaries}\label{sec2}
In this section, we introduce some basic definitions and concepts of sublinear expectation theory, in particular some properties of $G$-normal distributions. Readers interested in sublinear expectation theory can refer to \cite{peng2019nonlinear} for more details, such as the law of large numbers and the central limit theorem under the framework of sublinear expectations.

Let $\Omega$ be a given set as the sample space and let $\mathcal{H}$ be a linear space of real-valued functions on $\Omega$ as the space of random variables. We assume that if $X_1,\cdots,X_n\in\mathcal{H}$, then $\varphi(X_1,\cdots,X_n)\in\mathcal{H}$ for all $\varphi$ in $C(\mathbb{R}^n)$, where $C(\mathbb{R}^n)$ denotes the space of continuous functions on $\mathbb{R}^n$. $C_{l,Lip}(\mathbb{R}^n)$ denotes the space of functions $\varphi$ on $\mathbb{R}^n$ satisfying
\begin{align}
|\varphi(\mathbf{x})-\varphi(\mathbf{y})|\leq C(1+|\mathbf{x}|^m+|\mathbf{y}|^m)|\mathbf{x}-\mathbf{y}| \ \text{for}\ \mathbf{x},\mathbf{y}\in\mathbb{R}^n,\notag
\end{align}
where the constant $C>0$ and the integer $m\in\mathbb{N}$ depend on $\varphi$.
\begin{definition}
	A sublinear expectation $\mathbb{E}:\mathcal{H}\rightarrow\bar{\mathbb{R}}:=[-\infty,\infty]$ is a functional on $\mathcal{H}$ satisfying:
	\begin{enumerate}
		\item{\textit{Monotonicity}: $\mathbb{E}[X]\geq \mathbb{E}[Y]$, if $X\geq Y$.}
		\item{\textit{Constant preserving}: $\mathbb{E}[c]=c,\forall c \in \mathbb{R}$.}
		\item{\textit{Sub-additivity}: $\mathbb{E}[X+Y]\leq \mathbb{E}[X]+\mathbb{E}[Y]$ whenever $\mathbb{E}[X]+\mathbb{E}[Y]$ is not of the form $+\infty-\infty$ or $-\infty+\infty$.}
		\item{\textit{Positive homogeneity}: $\mathbb{E}[\lambda X]=\lambda \mathbb{E}[X]$, for all $\lambda\geq0$.}
	\end{enumerate}
\end{definition}
Here, $0\cdot\infty:=0$. The triplet $(\Omega,\mathcal{H},\mathbb{E})$ is called a sublinear expectation space. If the inequality in 3) takes equality, then $\mathbb{E}$ is a linear expectation and $(\Omega,\mathcal{H},\mathbb{E})$ becomes a linear expectation space. If only 1) and 2) are satisfied, then $\mathbb{E}$ is called a nonlinear expectation, and $(\Omega,\mathcal{H},\mathbb{E})$ is called a nonlinear expectation space. $\mathbf{X}=(X_1,\cdots,X_n)$ is said to be an $n$-dimensional random vector on $(\Omega,\mathcal{H},\mathbb{E})$ if each coordinate of $\mathbf{X}$ is on $(\Omega,\mathcal{H},\mathbb{E})$.

We use the following definitions of identical distributions and independence under sublinear expectations in Gu and Zhang \cite{GZ2024}.
\begin{definition}
	Let $\mathbf{X}_1$ and $\mathbf{X}_2$ be two $n$-dimensional random vectors defined on $(\Omega_1,\mathcal{H}_1,\mathbb{E}_1)$ and $(\Omega_2,\mathcal{H}_2,\mathbb{E}_2)$, respectively. We say they are identically distributed, denoted by $\mathbf{X}_1\overset{d}{=}\mathbf{X}_2$, if
	\begin{align}
	\mathbb{E}_1[\varphi(\mathbf{X}_1)]=\mathbb{E}_2[\varphi(\mathbf{X}_2)],\ \forall\varphi\in C_{l,Lip}(\mathbb{R}^n).\notag
	\end{align}
\end{definition}
\begin{definition}\label{independence}
	Let $\mathbf{Y}_1$ be an $m$-dimensional random vector on $(\Omega,\mathcal{H},\mathbb{E})$ and $\mathbf{Y}_2$ be an $n$-dimensional random vector on $(\Omega,\mathcal{H},\mathbb{E})$. We say $\mathbf{Y}_2$ is independent of $\mathbf{Y}_1$ if
	\begin{align}
	\mathbb{E}[\varphi(\mathbf{Y}_1,\mathbf{Y}_2)]=\mathbb{E}[\mathbb{E}[\varphi(\mathbf{y}_1,\mathbf{Y}_2)]|_{\mathbf{y}_1=\mathbf{Y}_1}]\notag
	\end{align}
	for $\varphi\in C_{l,Lip}(\mathbb{R}^{m+n})$ such that $\hat{\varphi}(\mathbf{y}_1):=\mathbb{E}[|\varphi(\mathbf{y}_1, \mathbf{Y}_2)|]<\infty$ for all $\mathbf{y}_1\in\mathbb{R}^m$, $\mathbb{E}[\varphi(\mathbf{y}_1,\mathbf{Y}_2)]|_{\mathbf{y}_1=\mathbf{Y}_1}\in\mathcal{H}$, $\hat{\varphi}(\mathbf{Y}_1)\in\mathcal{H}$ and $\mathbb{E}[\hat{\varphi}(\mathbf{Y}_1)]<\infty$.
\end{definition}

A sequence of random variables $\{X_k\}_{k\geq1}$ is called i.i.d., if $X_{k+1}\overset{d}{=}X_k$ and $X_{k+1}$ is independent of $(X_1,\cdots,X_{k})$ for any $k=1,2,\cdots$. It is important to observe that $\mathbf{Y}_2$ is independent of $\mathbf{Y}_1$ does not in general imply that $\mathbf{Y}_1$ is independent of $\mathbf{Y}_2$ under sublinear expectations, and one can refer to Example 1.3.15 in \cite{peng2019nonlinear} for more details.

Next, we state the concept of $G$-normal distribution under sublinear expectations. 

\begin{definition}
	Let $G(x):=\frac{1}{2}(\overline{\sigma}^2x^+-\underline{\sigma}^2x^-)$, $x\in \mathbb{R}$, where $\underline{\sigma}$ and $\overline{\sigma}$ are two given constants satisfying $0\leq \underline{\sigma}\leq \overline{\sigma}$. A random variable $\xi$ defined on a sublinear expectation space $(\Omega,\mathcal{H},\mathbb{E})$ is said to follow a $G$-normal distribution with the lower variance $\underline{\sigma}^2$ and the upper variance $\overline{\sigma}^2$, denoted as $N(\{0\}\times[\underline{\sigma}^2,\overline{\sigma}^2])$, if for any $\varphi\in C_{l,Lip}(\mathbb{R})$, the function
	\begin{align}
	u(t,x):=\mathbb{E}[\varphi(x+\sqrt{t}\xi)],\ t\geq0, \ x\in\mathbb{R}\notag
	\end{align}
	is the unique viscosity solution of the following $G$-heat equation
	\begin{align}
	\left\{
	\begin{array}{ll}
	\partial_t u = G(\partial_{xx}^2 u),  \\
	u(0,x)=\varphi(x).
	\end{array}
	\right. \notag
	\end{align}
\end{definition}

\begin{remark}
	Another equivalent definition that a random variable $\xi$ on $(\Omega,\mathcal{H},\mathbb{E})$ is $G$-normally distributed is that 
	\begin{align}
	a\xi+b\overline{\xi}\overset{d}{=}\sqrt{a^2+b^2}\xi,\ \text{for}\  a,b\geq0,\notag
	\end{align}
	where $\overline{\xi}$ is an independent copy of $\xi$ and $G(x):=\frac{1}{2}\mathbb{E}[x\xi^2]$, $x\in \mathbb{R}$ (see Definition 2.2.4 and Example 2.2.14 of Peng \cite{peng2019nonlinear}).
\end{remark}

The notation of $G$-normal distribution is easy to be misunderstood as a normal distribution with uncertain variance, but in fact, we must not regard the $G$-normal distribution as a normal distribution with uncertain variance.

On the one hand, the $G$-normal distribution no longer has some properties of the normal distribution. For example, the odd order moment of the normal distribution is equal to 0, while the odd order moment of the $G$-normal distribution is strictly greater than 0 (see (ii) of Lemma 2 in \cite{HMS2012}). The random vector composed of $n$ independent normal distributions obeys the $n$-dimensional normal distribution, but the random vector composed of $n$ independent $G$-normal distributions does not obey the $G$-normal distribution (see Theorem 4.2 in \cite{EA2015}). One can refer to \cite{EA2015} to learn more about the differences between the normal distribution and the $G$-normal distribution.

On the other hand, the uncertainty of the $G$-normal distribution is generally greater than that of the normal distribution with uncertain variance, and the following lemma establishes this result.
\begin{lemma}((i) of Lemma 2 in \cite{HMS2012})
	Let $\xi$ be a $G$-normally distributed random variable on a sublinear expectation space $(\hat{\Omega},\hat{\mathcal{H}},\hat{\mathbb{E}})$ with the sublinear function $G_\xi$ on $\mathbb{R}$,
	\begin{align}
	G_\xi(x)=\frac{1}{2}(x^+-\sigma^2x^-), \ x\in \mathbb{R} \notag
	\end{align} 
	where $\sigma\in[0,1)$ is a given constant. Let $Z$ be a standard normal random variable on a probability space $(\tilde{\Omega},\tilde{\mathcal{F}},\tilde{P})$. Then for each $\varphi\in C_{l,Lip}(\mathbb{R})$, 
	\begin{align}
	\hat{\mathbb{E}}[\varphi(\xi)]\geq\sup_{\sigma\leq v\leq 1}E_{\tilde{P}}[\varphi(vZ)].\notag
	\end{align}
	
\end{lemma} 

Although there are great differences between the $G$-normal distribution and the normal distribution, to some extent, the normal distribution with uncertain variance can approach the $G$-normal distribution (see Theorem 3.4 in \cite{LK2018}).

The definition of the $G$-normal distribution is related to partial differential equations. However, it is generally difficult to compute the sublinear expectation of random variables derived from composing the $G$-normal distribution with some functions. Nevertheless, the sublinear expectation can be easily calculated in the following two cases.
\begin{lemma}\label{le:convex}(Proposition 2.2.15 of \cite{peng2019nonlinear})
	Let $X\overset{d}{=}N(\{0\}\times[\underline{\sigma}^2,\overline{\sigma}^2])$ be on a sublinear expectation space $(\Omega,\mathcal{H},\mathbb{E})$, where $\overline{\sigma}\geq\underline{\sigma}>0$ are two given constants. Then for each convex (resp. concave) function $\varphi\in C_{l,Lip}(\mathbb{R})$, we have
	\begin{align}
	\mathbb{E}[\varphi(X)]
	&=\frac{1}{\sqrt{2\pi\overline{\sigma}^2}}\int_{-\infty}^{\infty}\varphi(y)\exp\Big(-\frac{y^2}{2\overline{\sigma}^2}\Big)dy,\notag\\
	&\bigg(resp. =\frac{1}{\sqrt{2\pi\underline{\sigma}^2}}\int_{-\infty}^{\infty}\varphi(y)\exp\Big(-\frac{y^2}{2\underline{\sigma}^2}\Big)dy\bigg).\notag
	\end{align}
\end{lemma}

We need the following estimation of the solution to $G$-heat equation, which will be used in Section \ref{sec4}.
\begin{lemma}\label{le:estimate}(Lemma 4.4 in \cite{HL2026})
	Let $G(x):=\frac{1}{2}(\overline{\sigma}^2x^+-\underline{\sigma}^2x^-)$, $x\in \mathbb{R}$, where $\underline{\sigma}$ and $\overline{\sigma}$ are two given constants satisfying $0< \underline{\sigma}\leq \overline{\sigma}$. Suppose $u_\beta$ is the solution of the following PDE defined on $[0,\infty)\times\mathbb{R}$,
	\begin{align}\label{parabolic PDE}
	\left\{
	\begin{array}{ll}
	\partial_t u_\beta(t,x) = G(\partial_{xx}^2 u_\beta(t,x)),  \\
	u_\beta(0,x)=\exp\Big(-\frac{\beta|x-a|^2}{2\overline{\sigma}^2}\Big),
	\end{array}
	\right. \notag
	\end{align}
	where $\beta>0$ and $a\in\mathbb{R}$. Then we have
	\begin{align}
	u_\beta(t,x)\leq(1+\beta t)^{-\rho}\exp\bigg(-\frac{\beta|x-a|^2}{2(1+\beta t)\overline{\sigma}^2}\bigg),\notag
	\end{align}
	where $\rho:=\frac{\underline{\sigma}^2}{2\overline{\sigma}^2}$.
\end{lemma}

\section{Problem Formulation}\label{sec3}
Let $(\Omega,\mathcal{H},\mathbb{E})$ be a sublinear expectation space, and let $\mathcal{P}$ be a family of probability measures on $\sigma(\mathcal{H})$, where $\sigma(\mathcal{H})$ is the smallest $\sigma$-algebra generated by $\mathcal{H}$. We suppose that $\mathbb{E}$ is an upper expectation with respect to $\mathcal{P}$, i.e., 
\begin{align}
\mathbb{E}[X]=\sup_{P\in\mathcal{P}}E_P[X],\  X\in\mathcal{H},\notag
\end{align}
which means the uncertainty of $\mathbb{E}$ is characterized by $\mathcal{P}$.

The energy detection problem we consider in this paper is formulated as a binary hypothesis-testing problem with the following two hypotheses:
\begin{align}
\begin{split}
&\mathcal{H}_0:Y_i=\epsilon_iX_i+Z_i,\ \  i=1,\cdots,n,\\
&\mathcal{H}_1:Y_i=Z_i,\ \ \ \ \ \ \ \ \ \ \ i=1,\cdots,n,\notag
\end{split}
\end{align}
where $n$ is the number of samples collected during the observations of the channel, the random variables $\{X_i\}_{i=1}^n$ stand for the signals that may exist in the channel, $\{\epsilon_i\}_{i=1}^n$ are the channel fading random variables, $\{Z_i\}_{i=1}^n$ are the noise random variables in the channel, and $\{Y_i\}_{i=1}^n$ are the output random variables.

Now, we state our assumptions for the model of the energy detection problem under the framework of sublinear expectations. The noise random variables $\{Z_i\}_{i=1}^n$ are identically distributed and follow a $G$-normal distribution $N(\{0\}\times[\underline{\sigma}^2,\overline{\sigma}^2])$, where $0<\underline{\sigma}\leq\overline{\sigma}$ are given constants. $\{\epsilon_i\}_{i=1}^n$ share the common deterministic distribution under $\mathbb{E}$, and we consider the case when the fading is either a constant or follows a Rayleigh, Rician, or Nakagami distribution. Constant fading is applicable to the case where the transmitter, receiver, and environment are relatively static, such as fixed wireless access. Rayleigh fading applies to multipath fading environments where there is no direct path. In such environments, the signal reaches the receiver through multiple paths, as seen in urban mobile communication scenarios, where buildings and other obstacles cause multipath reflections, and no direct path exists. The Rician distribution is applicable to environments where a direct line-of-sight path exists, such as in rural areas and satellite communication. As a general distribution, the Nakagami distribution can be adjusted through its parameters to accommodate various fading conditions. If a signal is present in the channel, we assume that each $X_i^2$ lies in $[\underline{\sigma}_X^2,\overline{\sigma}_X^2]$ for $i=1,\cdots,n$, where $0<\underline{\sigma}_X\leq\overline{\sigma}_X$ are two given constants. We suppose that each observation is independent of the previous ones. During each observation, the noise $Z_i$ interferes with the existing signal in an additive way, while $\epsilon_i$ interferes in a multiplicative way. The noise and channel fading interfere with signals in different ways, and we assume the noise is independent of channel fading during each observation. Moreover, we assume that they are independent of the input signal during each observation. Hence, all these mean $\{X_1,\epsilon_1,Z_1,X_2,\epsilon_2,Z_2,\cdots,X_n,\epsilon_n,Z_n\}$ is an independent sequence under the sublinear expectation $\mathbb{E}$.
\par In practical energy detection problems, two types of errors can occur: missed detections and false alarms. A missed detection occurs when a signal is present in the channel, and the detector selects hypothesis $\mathcal{H}_1$, which may result in harmful interference to users of the channel. On the other hand, a false alarm occurs when the channel is idle, and the detector selects hypothesis $\mathcal{H}_0$, which may reduce the utilization of the channel. The probability in our model is uncertain, and we use the supremum to dominate the error probabilities uniformly. For the convenience of expression, we use $X$ to denote the signal random variable. Influenced by the uncertainty of the input signals, we consider \textit{the maximum of the upper probability of missed detections}, i.e., $\underset{|X|\in[\underline{\sigma}_X,\overline{\sigma}_X]}{\sup} \underset{P\in\mathcal{P}}{\sup} P(\text{choosing}\ \mathcal{H}_1)_{\mathcal{H}_0}$, and \textit{the minimum of the upper probability of missed detections}, i.e., $\underset{|X|\in[\underline{\sigma}_X,\overline{\sigma}_X]}{\inf} \underset{P\in\mathcal{P}}{\sup} P(\text{choosing}\ \mathcal{H}_1)_{\mathcal{H}_0}$. We also consider \textit{the upper probability of false alarms}, i.e., $\underset{P\in\mathcal{P}}{\sup}P(\text{choosing}\ \mathcal{H}_0)_{\mathcal{H}_1}$, which is unaffected by the uncertainty of the input signals.

To balance both error rates, a reasonable approach is to first dominate the missed detection rate and then adjust the sample size to dominate the false alarm rate. By using nonlinear expectation theory, our model offers a unified approach to handle both signal and noise uncertainties, which are not accounted for by traditional energy detection models.

\section{Theoretical Analysis}\label{sec4}
We take $\sum_{i=1}^{n}Y_i^2$ as the test statistic, and our statistical test $\Psi$ for the binary hypothesis-testing
problem has the form as follows:
\begin{align}
\Psi:\ \text{If}\ \sum_{i=1}^{n}Y_i^2\leq\lambda,\ \text{then reject the null hypothesis.}\notag
\end{align}
Here $\lambda$ is the threshold to be determined by the level of missed detections and false alarms. The following theorem gives a statistical test for the binary hypothesis-testing problem. 
\begin{theorem}\label{missed detections}
	For arbitrary given $\beta>0$, $p>0$, $\Psi$ with 
	\begin{align}
	\lambda=-\frac{2\overline{\sigma}^2n}{\beta}\ln\bigg(\mathbb{E}\bigg[\exp\bigg\{-\frac{\beta\underline{\sigma}_X^2}{2\overline{\sigma}^2(1+\beta)}\epsilon_1^2\bigg\}\bigg]\bigg)+\frac{\underline{\sigma}^2n}{\beta}\ln(1+\beta)+\frac{2\overline{\sigma}^2}{\beta}\ln p\notag
	\end{align}
	is a statistical test with the maximum of the upper probability of missed detections not exceeding $p$.
\end{theorem}
\begin{proof}
	Note that the maximum of the upper probability of missed detections is
	\begin{align}
	\sup_{|X|\in[\underline{\sigma}_X,\overline{\sigma}_X]}\sup_{P\in\mathcal{P}}P\bigg(\sum_{i=1}^{n}(\epsilon_iX_i+Z_i)^2\leq\lambda\bigg),\notag 
	\end{align} 
	and for any $\beta>0$,
	\begin{align*}
	\sup_{P\in\mathcal{P}}P\bigg(\sum_{i=1}^{n}(\epsilon_iX_i+Z_i)^2\leq\lambda\bigg)
	&\leq\exp\bigg\{\frac{\beta\lambda}{2\overline{\sigma}^2}\bigg\}\mathbb{E}\bigg[\exp\bigg\{\frac{\sum_{i=1}^{n}-\beta(\epsilon_iX_i+Z_i)^2}{2\overline{\sigma}^2}\bigg\}\bigg]\\
	&=\exp\bigg\{\frac{\beta\lambda}{2\overline{\sigma}^2}\bigg\}\prod_{i=1}^{n}\mathbb{E}\bigg[\exp\bigg\{-\frac{\beta(\epsilon_iX_i+Z_i)^2}{2\overline{\sigma}^2}\bigg\}\bigg].
	\end{align*}
	By Definition \ref{independence} (the independence) and Lemma \ref{le:estimate}, for each $i=1,\cdots,n$, we have
	\begin{align*}
	\mathbb{E}\bigg[\exp\bigg\{-\frac{\beta(\epsilon_iX_i+Z_i)^2}{2\overline{\sigma}^2}\bigg\}\bigg]
	&=\mathbb{E}\bigg[\mathbb{E}\bigg[\exp\bigg\{-\frac{\beta(a+Z_i)^2}{2\overline{\sigma}^2}\bigg\}\bigg]\bigg|_{a=\epsilon_iX_i}\bigg]\\
	&\leq(1+\beta)^{-\rho}\mathbb{E}\bigg[\exp\bigg\{-\frac{\beta\epsilon_i^2X_i^2}{2\overline{\sigma}^2(1+\beta)}\bigg\}\bigg]\\
	&\leq(1+\beta)^{-\rho}\mathbb{E}\bigg[\exp\bigg\{-\frac{\beta\epsilon_i^2\underline{\sigma}_X^2}{2\overline{\sigma}^2(1+\beta)}\bigg\}\bigg]\\
	&=(1+\beta)^{-\rho}\mathbb{E}\bigg[\exp\bigg\{-\frac{\beta\epsilon_1^2\underline{\sigma}_X^2}{2\overline{\sigma}^2(1+\beta)}\bigg\}\bigg],
	\end{align*}
	where the last equality is due to the identically distributed property of $\{\epsilon_i\}_{i=1}^n$, and $\rho=\frac{\underline{\sigma}^2}{2\overline{\sigma}^2}$.
	\par Hence, 
	\begin{equation}\label{1}
	\sup_{|X|\in[\underline{\sigma}_X,\overline{\sigma}_X]}\sup_{P\in\mathcal{P}}P\bigg(\sum_{i=1}^{n}(\epsilon_iX_i+Z_i)^2\leq\lambda\bigg)\leq\exp\bigg\{\frac{\beta\lambda}{2\overline{\sigma}^2}\bigg\}(1+\beta)^{-\rho n}\bigg(\mathbb{E}\bigg[\exp\bigg\{-\frac{\beta\epsilon_1^2\underline{\sigma}_X^2}{2\overline{\sigma}^2(1+\beta)}\bigg\}\bigg]\bigg)^n.
	\end{equation}
	Then, letting the right side of \eqref{1} be equal to the given $p>0$, we obtain
	\begin{align}
	\lambda=-\frac{2\overline{\sigma}^2n}{\beta}\ln\bigg(\mathbb{E}\bigg[\exp\bigg\{-\frac{\beta\underline{\sigma}_X^2}{2\overline{\sigma}^2(1+\beta)}\epsilon_1^2\bigg\}\bigg]\bigg)+\frac{\underline{\sigma}^2n}{\beta}\ln(1+\beta)+\frac{2\overline{\sigma}^2}{\beta}\ln p,\notag
	\end{align}
	which leads to 
	\begin{align}
	\sup_{|X|\in[\underline{\sigma}_X,\overline{\sigma}_X]}\sup_{P\in\mathcal{P}}P\bigg(\sum_{i=1}^{n}(\epsilon_iX_i+Z_i)^2\leq\lambda\bigg)\leq p.\notag
	\end{align}
\end{proof}

Based on the proof of Theorem \ref{missed detections}, we can also derive the following corollary regarding the minimum of the upper probability of missed detections.

\begin{corollary}\label{missed detection 1}
	For statistical test $\Psi$ and arbitrary given $\beta>0$, there is
	\begin{align}
	\underset{|X|\in[\underline{\sigma}_X,\overline{\sigma}_X]}{\inf} \underset{P\in\mathcal{P}}{\sup} P\bigg(\sum_{i=1}^{n}(\epsilon_iX_i+Z_i)^2\leq\lambda\bigg) \leq \exp\bigg\{\frac{\beta\lambda}{2\overline{\sigma}^2}\bigg\}(1+\beta)^{-\rho n}\bigg(\mathbb{E}\bigg[\exp\bigg\{-\frac{\beta\epsilon_1^2\overline{\sigma}_X^2}{2\overline{\sigma}^2(1+\beta)}\bigg\}\bigg]\bigg)^n,\nonumber
	\end{align}
	where $\rho=\frac{\underline{\sigma}^2}{2\overline{\sigma}^2}$.
\end{corollary}

The above Theorem \ref{missed detections} and Corollary \ref{missed detection 1} provide the upper bound estimations of the maximum and the minimum of the upper probability of missed detections, and give a method for selecting $\lambda$ based on the given target missed detections probability. Subsequently, we give the estimation of the upper probability of false alarms for the statistical test $\Psi$ given in Theorem \ref{missed detections}.

\begin{theorem}
	For the threshold $\lambda$ in Theorem \ref{missed detections}, if there is $N>0$ such that $\lambda> n\overline{\sigma}^2$ for all $n\geq N$, then the upper probability of false alarms of $\Psi$ has the estimation that
	\begin{align}
	\sup_{P\in\mathcal{P}}P\bigg(\sum_{i=1}^{n}Z_i^2\geq\lambda\bigg)\leq \exp\bigg\{-\frac{d_{\beta}}{2\overline{\sigma}^2}\bigg\}\cdot\exp\bigg\{\bigg(\frac{1}{2}-\frac{1}{2}\ln\bigg(\frac{\overline{\sigma}^2}{k_{\beta}}\bigg)-\frac{k_{\beta}}{2\overline{\sigma}^2}\bigg)n\bigg\},\ n\geq N,\notag
	\end{align}
	where 
	\begin{align*}
	k_{\beta}&:=-\frac{2\overline{\sigma}^2}{\beta}\ln\bigg(\mathbb{E}\bigg[\exp\bigg\{-\frac{\beta\underline{\sigma}_X^2}{2\overline{\sigma}^2(1+\beta)}\epsilon_1^2\bigg\}\bigg]\bigg)+\frac{\underline{\sigma}^2}{\beta}\ln(1+\beta),\\
	d_{\beta}&:=\frac{2\overline{\sigma}^2}{\beta}\ln p.
	\end{align*}
\end{theorem}
\begin{proof}For $t\in(0,\frac{1}{2\overline{\sigma}^2})$,
	\begin{align}
	\sup_{P\in\mathcal{P}}P\bigg(\sum_{i=1}^{n}Z_i^2\geq \lambda\bigg)
	&\leq\inf_{0<t<\frac{1}{2\overline{\sigma}^2}}\frac{\mathbb{E}[e^{t\sum_{i=1}^{n}Z_i^2}]}{e^{t\lambda}}\nonumber\\
	&=\inf_{0<t<\frac{1}{2\overline{\sigma}^2}}\frac{\lim_{j\rightarrow\infty}\mathbb{E}[e^{t\sum_{i=1}^{n}((-j)\vee Z_i\wedge j)^2}]}{e^{t\lambda}}\nonumber\\
	&=\inf_{0<t<\frac{1}{2\overline{\sigma}^2}}\frac{\lim_{j\rightarrow\infty}(\mathbb{E}[e^{t((-j)\vee Z_1\wedge j)^2}])^n}{e^{t\lambda}}\nonumber\\
	&=\inf_{0<t<\frac{1}{2\overline{\sigma}^2}}\frac{(\mathbb{E}[e^{tZ_1^2}])^n}{e^{t\lambda}}\nonumber\\
	&=\inf_{0<t<\frac{1}{2\overline{\sigma}^2}}\frac{\bigg(\mathbb{E}\bigg[\sum_{k=0}^{\infty}\frac{(tZ_1^2)^k}{k!}\bigg]\bigg)^n}{e^{t\lambda}}\nonumber\\
	&\leq\inf_{0<t<\frac{1}{2\overline{\sigma}^2}}\frac{\bigg(\sum_{k=0}^{\infty}\frac{\mathbb{E}[(tZ_1^2)^k]}{k!}\bigg)^n}{e^{t\lambda}}.\nonumber
	\end{align}
	By Lemma \ref{le:convex}, we can calculate that
	\begin{align}
	\frac{\bigg(\sum_{k=0}^{\infty}\frac{\mathbb{E}[(tZ_1^2)^k]}{k!}\bigg)^n}{e^{t\lambda}}=\frac{1}{e^{t\lambda}(\sqrt{1-2\overline{\sigma}^2t})^n}=:g(t),\ 0<t<\frac{1}{2\overline{\sigma}^2}.\notag
	\end{align}
	By a simple calculation, we have
	\begin{align}
	g'(t)=\frac{2\overline{\sigma}^2\lambda t+n\overline{\sigma}^2-\lambda}{e^{t\lambda}(\sqrt{1-2\overline{\sigma}^2t})^{n+2}},\quad 0<t<\frac{1}{2\overline{\sigma}^2}.\notag
	\end{align}
	Then, for $n\geq N$,
	\begin{align}
	\inf_{0<t<\frac{1}{2\overline{\sigma}^2}}\frac{(\mathbb{E}[e^{tZ_1^2}])^n}{e^{t\lambda}}
	&\leq \exp\bigg(-\frac{n}{2}\ln\bigg(\frac{n\overline{\sigma}^2}{\lambda}\bigg)-\frac{\lambda}{2\overline{\sigma}^2}+\frac{n}{2}\bigg)\nonumber\\
	&\leq\exp\bigg\{-\frac{d_{\beta}}{2\overline{\sigma}^2}\bigg\}\cdot\exp\bigg\{\bigg(\frac{1}{2}-\frac{1}{2}\ln\bigg(\frac{\overline{\sigma}^2}{k_{\beta}}\bigg)-\frac{k_{\beta}}{2\overline{\sigma}^2}\bigg)n\bigg\},\nonumber
	\end{align}
	provided that $\lambda=k_{\beta}n+d_{\beta}$.
\end{proof}

Having derived the estimation for the upper probability of false alarms, it is crucial to determine how to choose $\beta$ in order to make the estimation meaningful. In the following examples, we analyze under which condition there exists $\beta$ such that the upper probability of false alarms of the statistical test $\Psi$ exponentially decays with the increase of sample size $n$. This analysis is conducted for four different fading channel models where the channel fading is constant, or follows a Rayleigh distribution, a Rician distribution, or a Nakagami distribution. These models represent typical fading behaviors encountered in practical communication systems.

\begin{example}\label{ex:constant}(Constant fading)
	If $\epsilon_i=\epsilon$, $i=1,\cdots,n$, where $\epsilon>0$ is a constant, then
	\begin{align}
		k_{\beta}=\frac{\underline{\sigma}_X^2\epsilon^2}{1+\beta}+\frac{\underline{\sigma}^2}{\beta}\ln(1+\beta).\notag
	\end{align}
	It is clear that $k_{\beta}\in(0,\underline{\sigma}_X^2\epsilon^2+\underline{\sigma}^2)$ when $\beta>0$. By a simple analysis, it can be verified that if $\underline{\sigma}_X^2\epsilon^2+\underline{\sigma}^2>\overline{\sigma}^2$,  then there exists $\beta_0>0$ such that
	\begin{align}
		\frac{1}{2}-\frac{1}{2}\ln\bigg(\frac{\overline{\sigma}^2}{k_{\beta}}\bigg)-\frac{k_{\beta}}{2\overline{\sigma}^2}<0\ \text{for all}\ \beta\in(0,\beta_0).\nonumber
	\end{align}
	\par Therefore, if $\underline{\sigma}_X^2\epsilon^2+\underline{\sigma}^2>\overline{\sigma}^2$, then there exists $\beta>0$ such that
	\begin{align}\label{3}
		\lambda>n\overline{\sigma}^2\ \text{for sufficiently large} \ n\ \text{and}\ \frac{1}{2}-\frac{1}{2}\ln\bigg(\frac{\overline{\sigma}^2}{k_{\beta}}\bigg)-\frac{k_{\beta}}{2\overline{\sigma}^2}<0,
	\end{align}
	which means the upper probability of false alarms of $\Psi$ exponentially decays to zero as the sample size $n$ goes to infinity.  
\end{example}

\begin{example}(Rayleigh fading)
	If the channel fading follows a Rayleigh distribution under $\mathbb{E}$ with the probability density function
	\begin{align}
	f(x)=\frac{x}{\sigma^2}\exp\bigg\{-\frac{x^2}{2\sigma^2}\bigg\},\ x\geq0,\notag
	\end{align}
	where $\sigma$ is the parameter of the Rayleigh distribution, then
	\begin{align}
	\mathbb{E}\bigg[\exp\bigg\{-\frac{\beta\underline{\sigma}_X^2}{2\overline{\sigma}^2(1+\beta)}\epsilon_1^2\bigg\}\bigg]=\frac{\beta\overline{\sigma}^2+\overline{\sigma}^2}{\beta(\sigma^2\underline{\sigma}_X^2+\overline{\sigma}^2)+\overline{\sigma}^2},\notag
	\end{align}
	and
	\begin{align}
	k_{\beta}=\ln\bigg(\frac{\beta(\sigma^2\underline{\sigma}_X^2+\overline{\sigma}^2)+\overline{\sigma}^2}{\beta\overline{\sigma}^2+\overline{\sigma}^2}\bigg)\cdot\frac{2\overline{\sigma}^2}{\beta}+\frac{\underline{\sigma}^2}{\beta}\ln(1+\beta).\notag
	\end{align}
	It is clear that $k_{\beta}\rightarrow0$ as $\beta\rightarrow\infty$, and $k_{\beta}\rightarrow2\underline{\sigma}_X^2\sigma^2+\underline{\sigma}^2$ as $\beta\rightarrow0$. Similar to the analysis of Example \ref{ex:constant}, we have if $2\underline{\sigma}_X^2\sigma^2+\underline{\sigma}^2>\overline{\sigma}^2$, then there exists $\beta>0$ such that \eqref{3} holds.
\end{example}

\begin{example}(Rician fading)
	When the channel fading follows a Rician distribution under $\mathbb{E}$, then the probability density function is
	\begin{align}
	f(x)=\frac{x}{\sigma^2}\exp\bigg\{-\frac{x^2+v^2}{2\sigma^2}\bigg\}I_0\bigg(\frac{xv}{\sigma^2}\bigg),\ x\geq0,\notag
	\end{align}
	where $I_0(\cdot)$ is the modified Bessel function of the first kind with order zero, and $\sigma$ and $v$ are the parameters of the Rician distribution. By the integral equation of the generalized Marcum $Q$-function in \cite{sun2008}, we have
	\begin{align}
	\mathbb{E}\bigg[\exp\bigg\{-\frac{\beta\underline{\sigma}_X^2}{2\overline{\sigma}^2(1+\beta)}\epsilon_1^2\bigg\}\bigg]=\frac{1}{2\sigma^2A_{\beta}+1}\cdot\exp\bigg\{-\frac{A_{\beta}v^2}{2\sigma^2A_\beta+1}\bigg\},\notag
	\end{align}
	and
	\begin{align}
	k_{\beta}=\frac{2\overline{\sigma}^2}{\beta}\bigg(\ln(2\sigma^2A_{\beta}+1)+\frac{A_{\beta}v^2}{2\sigma^2A_\beta+1}\bigg)+\frac{\underline{\sigma}^2}{\beta}\ln(1+\beta),\notag
	\end{align}
	where $A_{\beta}:=\frac{\beta\underline{\sigma}_X^2}{2\overline{\sigma}^2(1+\beta)}$. It is clear that $k_{\beta}\rightarrow2\sigma^2\underline{\sigma}_X^2+v^2\underline{\sigma}_X^2+\underline{\sigma}^2$ as $\beta\rightarrow0$, and $k_{\beta}\rightarrow0$ as $\beta\rightarrow\infty$. By a similar analysis we know that if $2\sigma^2\underline{\sigma}_X^2+v^2\underline{\sigma}_X^2+\underline{\sigma}^2>\overline{\sigma}^2$, then there exists $\beta>0$ such that \eqref{3} holds.
\end{example}

\begin{example}(Nakagami fading)
	If the channel fading follows a Nakagami distribution under $\mathbb{E}$ with the probability density function
	\begin{align}
	f(x)=\frac{2}{\Gamma(m)}\bigg(\frac{m}{\Omega}\bigg)^mx^{2m-1}\exp\bigg\{-\frac{m}{\Omega}x^2\bigg\},\ x\geq0,\notag
	\end{align}
	where $m\geq\frac{1}{2}$ and $\Omega>0$ are parameters of the Nakagami distribution, then by a calculation we have
	\begin{align}
	\mathbb{E}\bigg[\exp\bigg\{-\frac{\beta\underline{\sigma}_X^2}{2\overline{\sigma}^2(1+\beta)}\epsilon_1^2\bigg\}\bigg]=\bigg(\frac{2m\overline{\sigma}^2(1+\beta)}{\beta\Omega\underline{\sigma}_X^2+2m\overline{\sigma}^2(1+\beta)}\bigg)^m,\notag
	\end{align}
	and 
	\begin{align}
	k_{\beta}=m\ln\bigg(\frac{\beta(2m\overline{\sigma}^2+\Omega\underline{\sigma}_X^2)+2m\overline{\sigma}^2}{2m\overline{\sigma}^2\beta+2m\overline{\sigma}^2}\bigg)\cdot\frac{2\overline{\sigma}^2}{\beta}+\frac{\underline{\sigma}^2}{\beta}\ln(1+\beta).\notag
	\end{align}
	It is clear that $k_{\beta}\rightarrow0$ as $\beta\rightarrow\infty$, and $k_{\beta}\rightarrow\Omega\underline{\sigma}_X^2+\underline{\sigma}^2$ as $\beta\rightarrow0$. Hence, if $\Omega\underline{\sigma}_X^2+\underline{\sigma}^2>\overline{\sigma}^2$, then there exists $\beta>0$ such that \eqref{3} holds. 
\end{example}

\section{Numerical Results}\label{sec5}
In this section, we present numerical simulations to validate the effectiveness of the results in this paper. In \cite{ML2013signal}, the researchers considered the signal uncertainty and used different distribution parameters to model various signals. The method for calculating the error probability is
\begin{equation}\label{eq:classical1}
\overline{P}(\gamma_0)=\int P(y)f_{\gamma}^{MG}(y)dy,
\end{equation}
where $P(y)$ is the target probability, such as the probability of missed detections or false alarms. $\gamma$ denotes the SNR, follows a distribution with the density function
\begin{equation}\label{eq:classical2}
f_{\gamma}^{MG}(y) = \frac{\kappa}{\sqrt{2\pi}\sigma_{\gamma}}e^{-\frac{(y-\gamma_0)^2}{2\sigma_{\gamma}^2}}, \quad y\geq0,
\end{equation}
where $\kappa$ is the normalization constant, $\gamma_0$ is the average SNR, and $\sigma_{\gamma}^2=\Gamma n^{-\chi}\gamma_0^2$. It has been verified that different values of $\Gamma$ and $\chi$ can model the detection error probabilities for different types of input signals, such as digital TV ($\Gamma=0.01,\chi=0.59$), DAB-T($\Gamma=0.055,\chi=0.051$), E-GSM($\Gamma=0.23,\chi=0.24$), analogical TV($\Gamma=0.29,\chi=0.17$), UMTS($\Gamma=0.7,\chi=0.23$). In practical applications, it may be necessary to determine the signal type or parameter in advance.

In this paper, we present a general framework for probability models with uncertainty, where signal type detection is not required. Referring to the conclusions of this paper, energy consumption associated with signal type detection can be reduced. It is important to note that the probabilities of missed detections are inversely related to the probabilities of false alarms. Therefore, if we impose a strict requirement on missed detections, i.e., limiting the maximum of the upper probability of missed detections $\underset{|X|\in[\underline{\sigma}_X,\overline{\sigma}_X]}{\sup} \underset{P\in\mathcal{P}}{\sup} P(\text{choosing}\ \mathcal{H}_1)_{\mathcal{H}_0}$ to be less than a threshold $p$, this corresponds to a ``conservative scenario'' for false alarms, which drives up the upper probability of false alarms. Conversely, if we set the minimum of the upper probability of missed detections $\underset{|X|\in[\underline{\sigma}_X,\overline{\sigma}_X]}{\inf} \underset{P\in\mathcal{P}}{\sup} P(\text{choosing}\ \mathcal{H}_1)_{\mathcal{H}_0}$ to be less than $p$, it corresponds to an ``aggressive scenario'' for false alarms, which drives down the upper probability of false alarms. A reasonable inference is that these two scenarios represent the upper and lower bounds of the traditional conclusion.

\begin{figure}[h]
	\centering
	\subfloat[]{\includegraphics[width=3.2in]{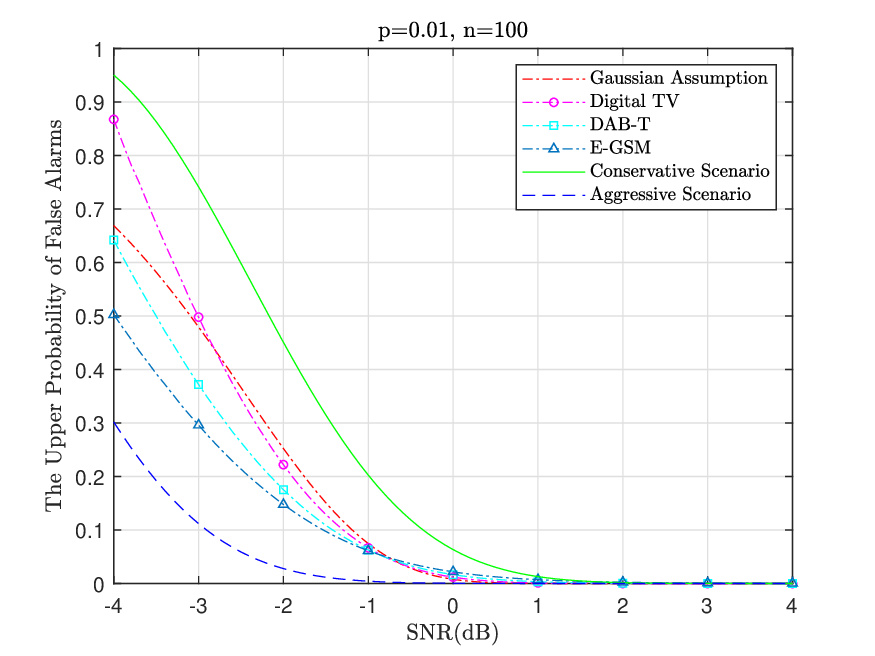}%
		\label{fig_1.1}}
	\hfil
	\subfloat[]{\includegraphics[width=3.2in]{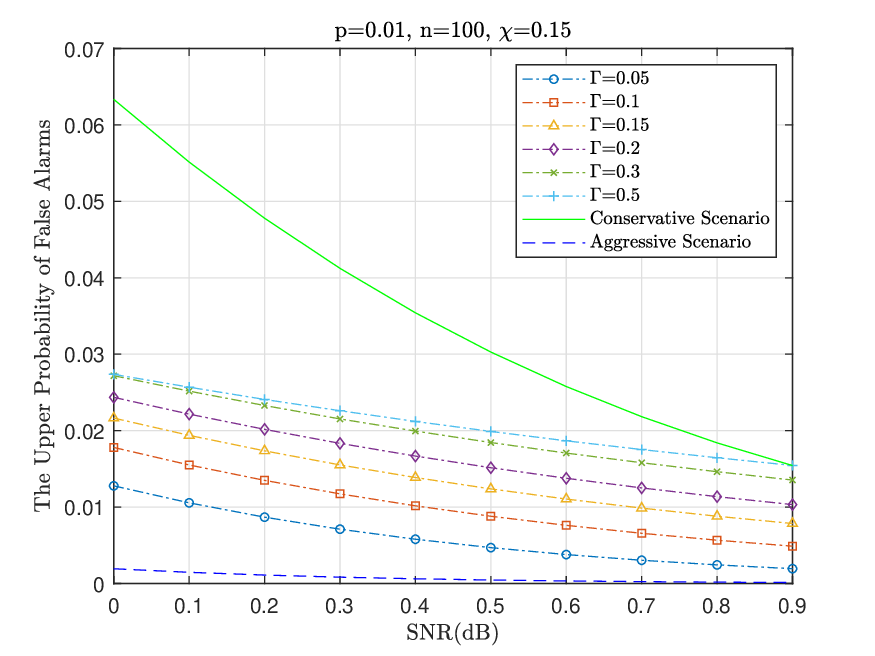}%
		\label{fig_1.2}}
	\captionsetup{font={scriptsize}}
	\caption{The upper probability of false alarms as a function of SNR: (a) various signals, (b) various values of $\Gamma$.}
	\label{fig_1}
\end{figure}

From a practical detection perspective, limiting the maximum of the upper probability of missed detections below a threshold can be interpreted as a worst-case performance metric for false alarms. This corresponds to the worst-case scenario in which false alarms may occur in a practical detection problem. Therefore, it represents the upper bound of the conclusion in classical robust detection. For a practical detection system, the actual error probability is expected to be no worse than this bound. In contrast, limiting the minimum of the upper probability of missed detections below a threshold can be interpreted as a favorable-case performance metric for false alarms, corresponding to the favorable-case scenario in which false alarms may occur in a practical detection problem. It represents the lower bound in classical robust detection. For a practical detection system, the actual error probability is typically no better than this bound.

In the following, we adopt a constant fading coefficient of $1$, which is also discussed in Example \ref{ex:constant}. Fig. \ref{fig_1} presents a comparison between our results and the traditional conclusions mentioned above. In Fig. \ref{fig_1}(a), we select the range of input signal values as $[\frac{1}{2}\sigma,3\sigma]$, which is based on the $3\sigma$-principle and $\sigma$ is determined by SNR\footnote{In communication theory, SNR is defined as the ratio of the input signal power to the noise power. This paper assumes that the noise has distributional uncertainty, and therefore, the SNR calculation used in this section is $\gamma^*:=\frac{1}{2}(\frac{\sigma^2}{\underline{\sigma}^2}+\frac{\sigma^2}{\overline{\sigma}^2})$. In addition, we always assume $\overline{\sigma}=\sqrt{2}\underline{\sigma}$ in this section.}. We plot the curves for both conservative (green solid line) and aggressive (blue dashed line) scenarios with respect to the upper probability of false alarms. Additionally, we use equations (\ref{eq:classical1}) and (\ref{eq:classical2}) to plot the corresponding curves for digital TV, DAB-T, and E-GSM, as well as the theoretical results assuming that both the input and noise follow a Gaussian distribution (i.e., the red dotted line with the legend ``Gaussian Assumption''). It can be observed that the conclusions of this paper encompass the curves corresponding to different signals in traditional approaches. For example, when the SNR ($\gamma^*$) is $-3$ dB and $p=0.01$, in the conservative scenario, the upper probability of false alarms can only be claimed to be less than $0.75$, while in the aggressive scenario, it can be claimed to be approximately $0.15$. This phenomenon is also reflected in the curves for different signal types. It can be observed that the values corresponding to different signal types can vary significantly, and relying solely on the Gaussian assumption curve makes it difficult to capture this difference. Fig. \ref{fig_1}(b) also shows that for different values of $\Gamma$, the results computed using equations (\ref{eq:classical1}) and (\ref{eq:classical2}) almost fall within the range of our conclusions.

\begin{remark}
	In Fig.\ref{fig_1}, the horizontal axis is labeled as SNR; the curve labeled ``Gaussian Assumption'' represents the corresponding SNR value; the curves labeled ``Conservative Scenario'' and ``Aggressive Scenario'' correspond to $\gamma^*$; and the curves labeled ``Digital TV'', ``DAB-T'', ``E-GSM'', and others correspond to $\gamma_0$.
\end{remark}

\begin{figure}[t]
	\centering
	\includegraphics[scale=0.55]{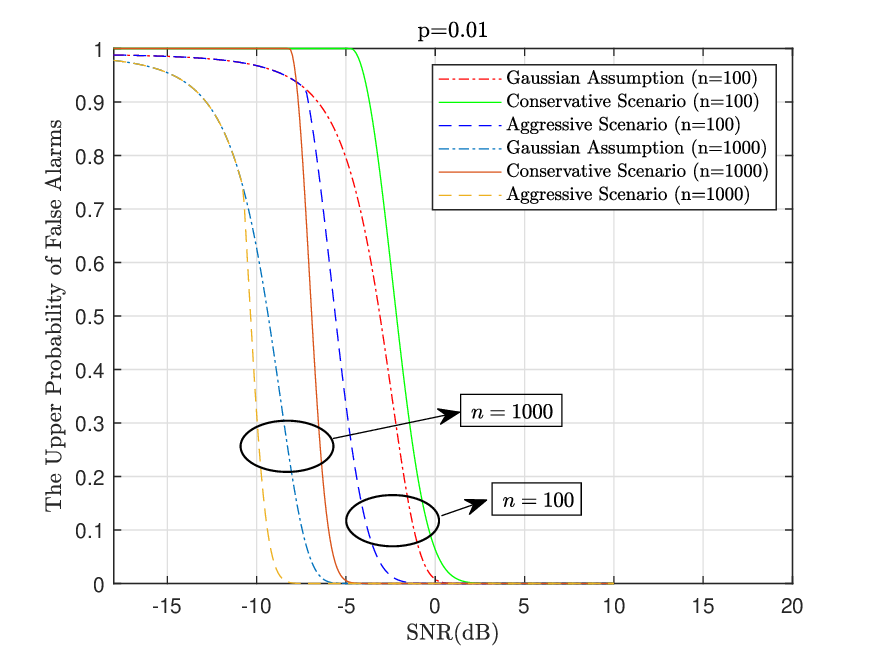}
	\captionsetup{font={scriptsize}}
	\caption{Theoretical performance of energy detection: $n=100$ and $n=1000$.}
	\label{fig_2}
\end{figure}

\begin{figure}[h]
	\centering
	\includegraphics[scale=0.55]{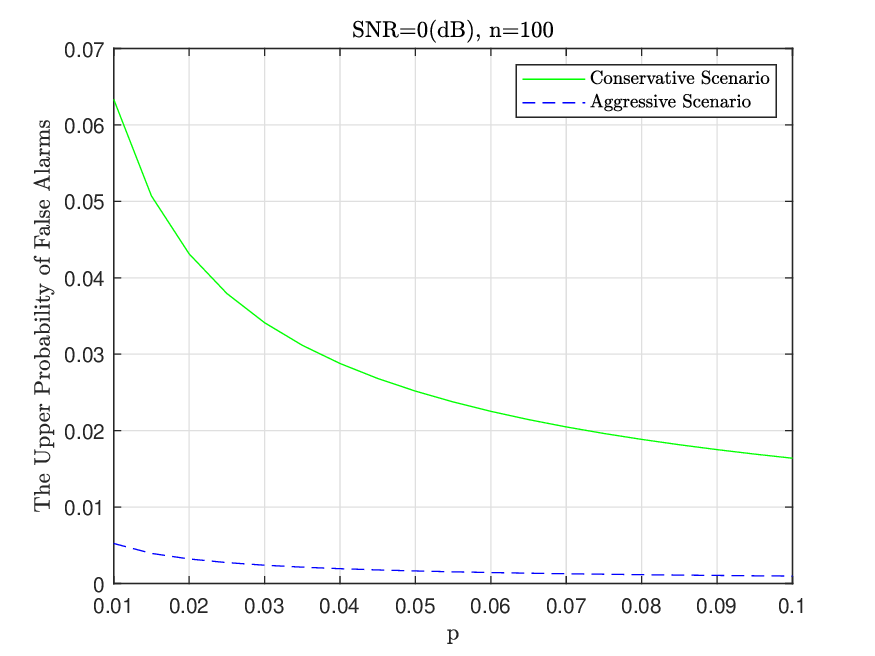}
	\captionsetup{font={scriptsize}}
	\caption{Receiver operating characteristic of an energy detector under distributional uncertainty.}
	\label{fig_3}
\end{figure}

Fig. \ref{fig_2} shows the comparison between the conclusions of this paper and the traditional theoretical results assuming that both the input and noise follow a Gaussian distribution (i.e., Gaussian assumption) for $n=100$ and $n=1000$.

Fig. \ref{fig_3} shows the joint variation curve of the upper probability of missed detections and the upper probability of false alarms. In the conservative scenario, the horizontal axis corresponds to the maximum of the upper probability of missed detections, while in the aggressive scenario, the horizontal axis corresponds to the minimum of the upper probability of missed detections. As illustrated in the figure, the numerical difference between the conservative and aggressive scenarios is quite significant, which highlights the need to consider both signal and distributional uncertainties.

\section{Conclusion}\label{sec6}
In this paper, we have explored the problem of energy detection in communication systems, focusing on addressing the challenge posed by distributional uncertainty and signal variety. We proposed a novel approach based on nonlinear expectation theory. Based on the uncertain-distribution hypothesis testing framework, and by utilizing the nonlinear expectation and properties of the $G$-normal distribution, we estimate the maximum and the minimum of the upper probability of missed detections, as well as the upper probability of false alarms. We also discussed the asymptotic properties of the conclusions under different channel fadings. Both numerical simulations and theoretical results indicate that the proposed model generalizes the classical theoretical analysis of energy detection. This approach lays the groundwork for energy robust detection strategies and provides a solid theoretical foundation for information-theoretic analysis in the presence of distributional ambiguity. Furthermore, the framework can be applied to a variety of scenarios in cognitive radio networks, particularly when the characteristics of primary signals and noise are not fully known in advance.

\section*{Acknowledgments}
The authors thank all colleagues and scholars who offered valuable discussions and suggestions in various stages of this research.

\end{document}